\documentclass[journal,12pt,onecolumn]{IEEEtran}
\IEEEoverridecommandlockouts
\usepackage{cite}
\usepackage{amsmath,amssymb,amsfonts,amsthm}
\usepackage{txfonts}
\usepackage{algorithmic}
\usepackage{graphicx}
\usepackage{caption}
\usepackage{textcomp}
\usepackage{xcolor}
\usepackage{tikz}
\usetikzlibrary{positioning, shapes.geometric}
\usepackage{setspace} 
\usepackage{comment}
\usepackage{url}
\usepackage{blkarray}
\newif\ifcomment
\commenttrue

\begin{document}

\newtheorem{theorem}{Theorem}
\newtheorem{definition}{Definition}
\newtheorem{corollary}{Corollary}
\newtheorem{lemma}{Lemma}
\newtheorem*{lemma*}{Lemma}
\newtheorem{example}{Example}
\newtheorem{remark}{Remark}

\title{Non-Existence of Some Function-Correcting Codes With Data Protection}

\author{\IEEEauthorblockN{ Charul Rajput\IEEEauthorrefmark{1}, B. Sundar Rajan\IEEEauthorrefmark{2}, Ragnar Freij-Hollanti\IEEEauthorrefmark{1}, Camilla Hollanti\IEEEauthorrefmark{1}}

\IEEEauthorblockA{\IEEEauthorrefmark{1}Department of Mathematics and Systems Analysis, Aalto University, Finland
    \\\{charul.rajput, ragnar.freij, camilla.hollanti\}@aalto.fi}
    \IEEEauthorblockA{\IEEEauthorrefmark{2}Department of Electrical Communication Engineering, Indian Institute of Science, Bengaluru, India
    \\\ bsrajan@iisc.ac.in}}

\maketitle

\begin{abstract}
In this paper, we consider the recently introduced concept of \emph{function-correcting codes (FCCs) with data protection}, which provide a certain level of error protection for the data and a higher level of protection for a desired function on the data. These codes are denoted by $(f\!:\!d_d,d_f)$-FCC, where $d_d$ is the minimum distance of the code and $d_f$ denotes the minimum distance between those codewords that correspond to different function values of a function $f:\mathbb{F}_q^k \to \mathrm{Im}(f)$, with $d_f \geq d_d$. We use a distance graph on a code based on the pairwise distances of its codewords, and show conditions under which a code cannot work as a \emph{strict} $(f\!:\!d_d,d_f)$-FCC, that is, code for which $d_f > d_d$. We then consider some well-known classes of codes, such as perfect codes and maximum distance separable (MDS) codes, and show that they cannot be used as \emph{strict} $(f\!:\!d_d,d_f)$-FCCs.
\end{abstract}
\begin{IEEEkeywords}
Error-correction, function-correcting codes, MDS codes, perfect codes, redundancy
\end{IEEEkeywords}
  \section{Introduction}
\label{intro}

Function-correcting codes (FCCs), introduced in \cite{LBWY2023}, are designed to protect specific functions or attributes of a message vector rather than the entire vector itself. In the standard communication setup, a transmitter sends a message vector over a noisy channel to a receiver who is interested only in the value of a function of the message. Unlike classical error-correcting codes (ECCs), which guarantee reliable recovery of the full message, FCCs impose distance constraints only between codewords corresponding to distinct function values. This targeted form of protection allows FCCs to achieve reliability for the function of interest with potentially lower redundancy, especially when the function has a significantly smaller image than its domain.

Since the introduction of FCCs, several works have studied their fundamental limits and constructions for various channel models and classes of functions. In \cite{LBWY2023}, the authors establish the basic framework of function-correcting codes by formally defining the FCC model and deriving upper and lower bounds on the redundancy required to protect a function. They also present explicit constructions for certain classes of functions, such as locally binary functions and the Hamming weight function, and show that some of these constructions are optimal. This work primarily considers communication over binary symmetric channels. Subsequent works extend the concept of FCCs to other channel models and settings. In particular, FCCs for symbol-pair read channels are studied in \cite{XLC2024} for the binary case and later generalized to $b$-symbol read channels in \cite{SSY2025}. In \cite{PR2025}, the authors focus on linear functions, where the kernel of the function induces a natural partition of the domain, and use this structure to construct FCCs and derive redundancy bounds. Further works derive improved bounds and constructions of FCCs for specific families of functions, including the Hamming weight function and locally bounded functions \cite{GXZZ2025,RRHH2025ITW25,VSS2025}. Redundancy bounds that are near-optimal up to logarithmic factors were derived in \cite{LS2025}, where tightness was also shown for sufficiently large field sizes. FCCs have also been investigated under alternative distance measures. In particular, Plotkin-type bounds and corresponding code constructions in the Lee metric were presented in \cite{VS2025,HUR2025}, while FCCs with homogeneous distance were studied in \cite{LL2025}. In addition, Plotkin-type bounds for FCCs over $b$-symbol read channels were obtained in \cite{SR2025}.

Across this line of work, FCCs are defined so as to protect only the function values, in the sense that distance constraints are imposed only between codewords corresponding to different function outputs. As a result, these codes do not provide any explicit guarantee on the level of error protection for the underlying data symbols themselves.

In many practical scenarios, however, it is desirable to simultaneously protect both the data and a function of the data, with different levels of protection, since certain attributes or functions of the data may be more important than the data itself. Motivated by this, the notion of \emph{function-correcting codes with data protection} was recently introduced in \cite{RRHH2025}. In this framework, an encoding is required to satisfy two distance constraints: a minimum distance $d_d$ that ensures error protection for the data itself, and a (possibly larger) minimum distance $d_f$ between codewords corresponding to distinct function values. When $d_f > d_d$, the code provides strictly stronger protection for the function than for the data, and such codes are referred to as \emph{strict} $(f\!:\!d_d,d_f)$-FCCs.

While constructions and redundancy bounds for FCCs with data protection have been explored in \cite{RRHH2025}, in this work we focus on the existence of strict FCCs for specific parameters. In particular, it is natural to ask whether well-known optimal families of classical error-correcting codes can be used to provide additional protection for function values beyond what they already offer for data. Addressing this question is the main focus of this paper.

In this work, we develop a graph-theoretic framework to study the feasibility of strict $(f\!:\!d_d,d_f)$-FCCs. To this end, we associate to any code a graph whose vertices are the codewords and whose edges are determined by pairwise Hamming distances. We show that the connectivity properties of these graphs impose strong restrictions on the achievable function distance $d_f$. In particular, if the distance graph of a code is sufficiently connected, then the code cannot serve as a strict FCC for any nontrivial function.

Using this framework, we derive general non-existence results for strict FCCs based on classical parameters of codes, such as the covering radius. As concrete consequences, we prove that two fundamental families of optimal error-correcting codes, perfect codes and MDS codes, cannot be used as strict $(f\!:\!d_d,d_f)$-FCCs for any nontrivial function. That is, these codes cannot provide stronger protection for function values than for the data itself. We further extend these results to other classes of codes for which the covering radius is known or can be bounded. 
This work is an extension of the full version \cite{RRHH2025}, with some added content, such as the $\alpha$-distance graph, which is a generalization of the minimum-distance graph given in the full version, and some new results, such as Lemma~\ref{adG3}.

\subsection{Contribution and paper organization}

The main contributions of this paper are summarized as follows.
\begin{itemize}
    \item We study the existence of \emph{strict} $(f\!:\!d_d,d_f)$-function-correcting codes, where the protection guaranteed for the function values is strictly stronger than that for the data.
    \item We introduce distance graphs of codes and show how the connectivity properties of these graphs impose fundamental limitations on the achievable function distance.
    \item Using this framework, we derive general non-existence results for strict $(f\!:\!d_d,d_f)$-FCCs in terms of classical code parameters, such as the covering radius.
    \item As concrete consequences, we prove that well-known optimal families of error-correcting codes, including perfect codes, MDS codes, and linear codes satisfying specific structural conditions, cannot serve as strict $(f\!:\!d,d_f)$-FCCs for any
nontrivial function.
    \item We further extend these non-existence results to other classes of codes for which the covering radius is known or can be bounded.
\end{itemize}

The remainder of the paper is organized as follows. Section~\ref{preliminaries} introduces basic definitions and concepts used throughout the paper. In Section~\ref{FCC_DP}, we recall the definition of function-correcting codes with data protection. Section~\ref{main} contains the main results of the paper, where we establish non-existence results for strict $(f\!:\!d_d,d_f)$-FCCs using distance-graph arguments. Finally, Section~\ref{conclusion} concludes the paper.

\noindent\textbf{Notation:}
The finite field of size $q$ is denoted by $\mathbb{F}_q$, and $\mathbb{F}_q^k$ denotes the $k$-dimensional vector space over $\mathbb{F}_q$. For vectors $u,v \in \mathbb{F}_q^n$, $d(u,v)$ denotes the Hamming distance between $u$ and $v$, and for a vector $u \in \mathbb{F}_q^n$, $\mathrm{wt}(u)$ denotes the Hamming weight of $u$. For integers $n$ and $r$, the binomial coefficient is denoted by
$
\binom{n}{r} \triangleq \frac{n!}{r!(n-r)!}.
$ 
Throughout the paper, vectors in $\mathbb{F}_q^k$ are sometimes written without commas for notational convenience, for example, $0110 \in \mathbb{F}_2^4$.

\section{Preliminaries}
\label{preliminaries}

This section presents some basic notions of codes \cite{L1992}, and the definition of function-correcting codes, based on the work in \cite{LBWY2023}.

A $q$-ary \emph{error-correcting code} (ECC) $\mathcal{C}$ of length $n$ is a subset of $\mathbb{F}_q^n$ with cardinality $M = |\mathcal{C}|$. The \emph{minimum distance} of $\mathcal{C}$, denoted by $d_{\min}(\mathcal{C})$ or simply $d$, is defined as
$$
d \triangleq \min_{\substack{c_1,c_2 \in \mathcal{C} \\ c_1 \neq c_2}} d(c_1,c_2).
$$
Such a code is commonly referred to as an $(n,M,d)$ code. 
For a vector $x \in \mathbb{F}_q^n$, the distance of $x$ from the code $\mathcal{C}$ is defined as
$$
d(x,\mathcal{C}) \triangleq \min_{c \in \mathcal{C}} d(x,c).
$$
More generally, for two nonempty subsets $A, B \subseteq \mathbb{F}_q^n$, the distance between the subsets $A$ and $B$ is defined as
$$
d(A, B) \triangleq \min_{a \in A, b \in B} d(a,b).
$$
\begin{definition}[Covering radius]
The \emph{covering radius} of a code $\mathcal{C}$, denoted by $R(\mathcal{C})$, is defined as
\[
R(\mathcal{C}) \triangleq \max_{x \in \mathbb{F}_q^n} d(x,\mathcal{C}).
\]
\end{definition}

Equivalently, $R(\mathcal{C})$ is the smallest integer such that the union of all Hamming balls of radius $R(\mathcal{C})$ centered at the codewords of $\mathcal{C}$ covers the entire space $\mathbb{F}_q^n$.
The covering radius measures the maximum distance from any vector in the ambient space to the nearest codeword.

\begin{definition}[Function-Correcting Codes]
Consider a function $f: \mathbb{F}_q^k \rightarrow Im(f)$. A systematic encoding $\mathcal{C}: \mathbb{F}_q^k \rightarrow \mathbb{F}_q^{k+r}$ is defined as an $(f, t)$-function correcting code (FCC) if, for any $u_1, u_2 \in \mathbb{F}_q^k$ such that $f(u_1) \neq f(u_2)$, the following condition holds: 
$$d(\mathcal{C}(u_1), \mathcal{C}(u_2)) \geq 2t+1.$$
\end{definition}

\section{Function-correcting codes with Data Protection}\label{FCC_DP}

In this section, we give the definition of function-correcting codes that simultaneously protect data and function values. For more details on such codes, we refer to \cite{RRHH2025}.

\begin{definition}
Consider a function $f: \mathbb{F}_q^k \rightarrow Im(f)$. An encoding $\mathfrak{C}_f: \mathbb{F}_q^k \rightarrow \mathbb{F}_q^{k+r}$ is defined as an $(f\!:d_d,d_f)$-function correcting code (FCC) if, 
\begin{itemize}
    \item for any $u_1, u_2 \in \mathbb{F}_q^k$, such that $u_1 \ne u_2$,
$$d(\mathfrak{C}_f(u_1), \mathfrak{C}_f(u_2)) \geq d_d,$$
\item 
for any $u_1, u_2 \in \mathbb{F}_q^k$ such that $f(u_1) \neq f(u_2)$, the following condition holds: 
$$d(\mathfrak{C}_f(u_1), \mathfrak{C}_f(u_2)) \geq d_f,$$
\end{itemize}
where $d_d$ and $d_f$ are two non-negative integers such that $d_d \le d_f$. 
\end{definition}

The following example illustrates one FCCs with data protection.

\begin{example}\label{ex:fcc_dd3_df5}
Consider the same parity function $f:\mathbb{F}_2^4\to\{0,1\}$,
\[
f(u)=\mathrm{wt}(u)\bmod 2 = u_1\oplus u_2\oplus u_3\oplus u_4,
\quad u=(u_1,u_2,u_3,u_4).
\]
We construct an $(f\!:\!d_d=3,d_f=5)$-FCC by using two step construction method given in \cite{RRHH2025}. By using a $[7,4,3]$ Hamming code
(with systematic encoding) in first step and a 2-bit  parity vectors that depend only on $f(u)$ in second step, we get the following code.

The parity bits of a $[7,4,3]$ Hamming code are
\[
p_1 = u_1\oplus u_2\oplus u_4,\
p_2 = u_1\oplus u_3\oplus u_4,\
p_3 = u_2\oplus u_3\oplus u_4,
\]
and the FCC parity bits are
\[
s_1=s_2=f(u)=u_1\oplus u_2\oplus u_3\oplus u_4.
\]
Now define the encoding $\mathfrak{C}_f:\mathbb{F}_2^4\to\mathbb{F}_2^9$ as
\[
\mathfrak{C}_f(u) = \big(u_1,u_2,u_3,u_4,\, p_1,p_2,p_3,\, s,s\big).
\]

\begin{center}
\begin{tabular}{ccccc}
\hline
$u\in\mathbb{F}_2^4$ 
& $f(u)$ 
& $(p_1p_2p_3)$ 
& $(s,s)$ 
& Codeword\\
\hline
0000 & 0 & 000 & 00 & 0000\,000\,00 \\
0001 & 1 & 111 & 11 & 0001\,111\,11 \\
0010 & 1 & 011 & 11 & 0010\,011\,11 \\
0011 & 0 & 100 & 00 & 0011\,100\,00 \\
0100 & 1 & 101 & 11 & 0100\,101\,11 \\
0101 & 0 & 010 & 00 & 0101\,010\,00 \\
0110 & 0 & 110 & 00 & 0110\,110\,00 \\
0111 & 1 & 001 & 11 & 0111\,001\,11 \\
1000 & 1 & 110 & 11 & 1000\,110\,11 \\
1001 & 0 & 001 & 00 & 1001\,001\,00 \\
1010 & 0 & 101 & 00 & 1010\,101\,00 \\
1011 & 1 & 010 & 11 & 1011\,010\,11 \\
1100 & 0 & 011 & 00 & 1100\,011\,00 \\
1101 & 1 & 100 & 11 & 1101\,100\,11 \\
1110 & 1 & 000 & 11 & 1110\,000\,11 \\
1111 & 0 & 111 & 00 & 1111\,111\,00 \\
\hline
\end{tabular}
\end{center}

As can be verified, for any $u\neq v$, we have
$d\big(\mathfrak{C}_f(u),\mathfrak{C}_f(v)\big)\ \ge\ 3.
$
Thus, $d_d=3$ holds. If $f(u)\neq f(v)$, then $s\neq s'$ and the last two parity bits differ in both positions, so
we have 
$d\big(\mathfrak{C}_f(u),\mathfrak{C}_f(v)\big)
\ge 3 + 2 = 5.$
Thus, $d_f=5$ holds. 
Therefore, $\mathfrak{C}_f$ is an $(f\!:\!3,5)$-FCC of length $9$ (redundancy $r=5$).
\end{example}

\begin{definition}[Optimal redundancy]
The \emph{optimal redundancy}, denoted by $r_f(k\!: \!d_d, d_f)$, is defined as the minimum value of $r$ for which there exists an $(f\!:d_d,d_f)$-FCC with an encoding function $\mathfrak{C}_f: \mathbb{F}_q^k \rightarrow \mathbb{F}_q^{k + r}$.
\end{definition}

\section{Non-existence of strict  $(f\!:\!d_d,d_f)$-FCC}\label{main}

We call an $(f\!:\!d_d,d_f)$-FCC a {\it strict  $(f\!:\!d_d,d_f)$-FCC} if $d_f > d_d$ since the case $d_f=d_d$ reduces to ECC.
In this section, we define the $\alpha$-distance graph for a code $\mathcal{C}$ and then using this we show that certain classes of codes cannot be used as strict $(f\!:\!d_d,d_f)$-FCCs. 

\begin{definition}[$\alpha$-distance graph]
Let $\mathcal{C}$ be a code. The \emph{$\alpha$-distance graph} of $\mathcal{C}$, denoted by $G_{\alpha}(\mathcal{C})$, is the graph whose vertex set is $\mathcal{C}$ and two distinct vertices $c_1, c_2 \in \mathcal{C}$ are adjacent if and only if 
$d(c_1, c_2) \leq \alpha,
$
where $\alpha\geq d_{\min}(\mathcal{C})$, and $d_{\min}(\mathcal{C})$ denotes the minimum distance of the code $\mathcal{C}$.
\end{definition}

In particular, the $\alpha$-distance graph is called the \emph{minimum-distance graph} if
$\alpha = d_{\min}(\mathcal{C})$, and is denoted by
$G(\mathcal{C}) = G_{d_{\min}}(\mathcal{C})$.
 Clearly, $G_{\alpha}(\mathcal{C})$ is a subgraph of
$G_{\alpha'}(\mathcal{C})$ whenever $\alpha \le \alpha'$. Consequently, if
$G_{\alpha}(\mathcal{C})$ is connected, then $G_{\alpha'}(\mathcal{C})$ is also connected. Moreover, if $d_{\max}$ denotes the maximum distance between any two codewords, then $G_{d_{\max}}(\mathcal{C})$ is a complete graph.

\begin{example}\label{ex:aDG}
Let $ \mathcal C=\{00000,\,00001,\,00010,\,01111,\,10111,$ $11111\}\subset \mathbb{F}_2^5.$
The $\alpha$-distance graph with $\alpha=2>d_{\min}(\mathcal C)$, i.e., $G_{2}(\mathcal{C})$ has edges only between pairs at distance less than or equal to 2. 

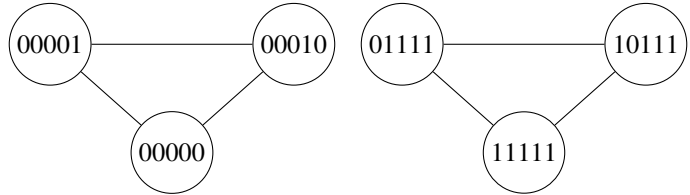
\begin{figure}[h]
\centering
\begin{tikzpicture}[scale=0.9,
  every node/.style={circle,draw,inner sep=2pt,minimum size=18pt}]
  \node (a) at (0,0)   {\small 00000};
  \node (b) at (-1.8,1.6) {\small 00001};
  \node (c) at (1.8,1.6)  {\small 00010};

  \node (d) at (5.2,0)   {\small 11111};
  \node (e) at (3.4,1.6) {\small 01111};
  \node (f) at (7,1.6) {\small 10111};

  \draw (a)--(b);
  \draw (a)--(c);
  \draw (b)--(c);

  \draw (d)--(e);
  \draw (d)--(f);
  \draw (e)--(f);

\end{tikzpicture}
\end{figure}
\end{example}

The following theorem illustrates how knowledge of the $\alpha$-distance graph can be used to determine whether a code can serve as a $(f\!:d_d,d_f)$-FCC with certain parameters.

\begin{theorem}\label{aDG1}
Let $\mathcal{C}$ be an $(n,q^k,d)$ code. If the $\alpha$-distance graph $G_{\alpha}(\mathcal{C})$ is a connected graph, then $\mathcal{C}$ cannot be an $(f\!:d,d_f)$-FCC for any $f:\mathbb{F}_q^k\to \text{Im}(f)$ with $|\text{Im}(f)| \geq 2$ and $d_f>\alpha$.
\end{theorem}

\begin{proof}
Assume $G_{\alpha}(\mathcal C)$ is connected and suppose that $\mathcal C$ serves as an $(f\!:d,d_f)$ with $d_f>\alpha$ for a function $f:\mathbb F_q^k\to \text{Im}(f)$ with $|\text{Im}(f)| \geq 2$. For each $a\in\mathrm{Im}(f)$, define the following set
$$
\mathcal C_a=\{c_u\in\mathcal C :\ f(u)=a\},
$$
where $c_u$ denotes the codeword corresponding to the message vector $u\in \mathbb{F}_q^k$. If $|\mathrm{Im}(f)|\ge2$ and $d_f>\alpha$, then by the $(f\!:d,d_f)$-FCC requirement, every $x\in\mathcal C_a$ and $y\in\mathcal C_b$ with $a\ne b$ must satisfy $d(x,y)\ge d_f>\alpha$. Hence, no edge of $G_{\alpha}(\mathcal C)$ exists between distinct $\mathcal C_a$ and $\mathcal C_b$, so the partition $\{\mathcal C_a : a \in \mathrm{Im}(f)\}$ separates the graph, contradicting the connectedness of $G_{\alpha}(\mathcal C)$. Therefore, $\mathcal{C}$ cannot be an $(f\!:d,d_f)$-FCC.
\end{proof}

Theorem \ref{aDG1} can easily be generalized for codes whose $\alpha$-distance graph has more than one connected component. 

\begin{theorem}\label{aDG2}
Let $\mathcal{C}$ be a $(n,q^k,d)$ code. If the $\alpha$-distance graph $G_{\alpha}(\mathcal{C})$ has $Q$ number of connected components, then $\mathcal{C}$ cannot be a $(f\!:d,d_f)$-FCC for any $f:\mathbb{F}_q^k\to \text{Im}(f)$ with $|\text{Im}(f)| \geq Q+1$ and $d_f>\alpha$. 
\end{theorem}

Since the connectivity of $G_{\alpha}(\mathcal{C})$ is preserved as $\alpha$
increases, when applying Theorems~\ref{aDG1} it is sufficient to
identify the smallest value of $\alpha$ for which $G_{\alpha}(\mathcal{C})$
becomes connected. This choice yields the strongest exclusion on the function
distance $d_f$, since any larger value of $\alpha$ results in weaker
constraints on the feasible values of $d_f$. Moreover, as $\alpha \ge d$, where $d$ is the minimum distance
of $\mathcal{C}$, establishing the connectivity of the minimum-distance graph
$G_{d}(\mathcal{C})$ yields the strongest possible restriction on the function
distance $d_f$.

The following lemma gives a condition for the $\alpha$-distance graph to be connected.
\begin{lemma}\label{adG3}
Let $C$ be a code with covering radius $R(C)$. Then the $\alpha$-distance graph $G_{\alpha}(C)$ with $\alpha \geq 2R(C) + 1$ is connected.
\end{lemma}

\begin{proof}
Let $C = A \cup B$ be a partition of $C$ into two non-trivial parts such that $d(A,B) = \delta$.
Let $a^* \in A$ and $b^* \in B$ be such that $d(a^*, b^*) = \delta$.
Choose a vector $x \in \mathbb{F}_q^n$ satisfying
\[
d(x, a^*) = \left\lceil \frac{\delta}{2} \right\rceil, \qquad
d(x, b^*) = \left\lfloor \frac{\delta}{2} \right\rfloor.
\]
Then,
\[
d(x, A) \le d(x, a^*) = \left\lceil \frac{\delta}{2} \right\rceil, \quad
d(x, B) \le d(x, b^*) = \left\lfloor \frac{\delta}{2} \right\rfloor.
\]
By the triangle inequality, for any $a \in A$, we have
\[
\delta \le d(a, b^*) \le d(a, x) + d(x, b^*) = d(a, x) + \left\lfloor \frac{\delta}{2} \right\rfloor,
\]
which gives $d(a, x) \ge \left\lceil \frac{\delta}{2} \right\rceil$.
Hence,
$
d(x, A) = \left\lceil \frac{\delta}{2} \right\rceil.$ Similarly, we can get $d(x, B) = \left\lfloor \frac{\delta}{2} \right\rfloor.$
Further by the definition of covering radius, we have
\[
R(C) \ge \max_{x' \in \mathbb{F}_q^n} d(x', C) \ge d(x, C)
= \min \{ d(x, A), d(x, B) \}
= \left\lfloor \frac{\delta}{2} \right\rfloor.
\]
Thus, $\delta \le 2R(C) + 1.$ 
This means the minimum distance between any two disconnected components is at most $2R(C) + 1$. Therefore, if we construct a graph $G_{2R(C)+1}(C)$ by connecting codewords at distance less than or equal to $2R(C) + 1$, then $G_{2R(C)+1}(C)$ is connected. Hence any $G_{\alpha}(C)$ with $\alpha \geq 2R(C) + 1$ is connected. 
\end{proof}

Using Theorem \ref{aDG1} and Lemma \ref{adG3}, we have the following result.
\begin{corollary}\label{aDG4}
Let $\mathcal{C}$ be a $(n,q^k,d)$ code with covering radius $R(\mathcal{C})$. Then $\mathcal{C}$ cannot be an $(f\!:d,d_f)$-FCC for any $f:\mathbb{F}_q^k\to \text{Im}(f)$ with $|\text{Im}(f)| \geq 2$ and $d_f>2R(C)+1$.
\end{corollary}

\subsection{Minimum-Distance Graphs of Perfect Codes and MDS Codes}

In this subsection, we show that the minimum-distance graph is connected for the two optimal families of ECCs: 1) perfect codes and 2) MDS codes.

\subsubsection{Perfect Codes}
The Hamming bound on the size of an error-correcting codes $\mathcal{C}$ with length $n$ and minimum distance $d$  is given in \cite{H1986} as follows.
$$M \leq \frac{q^n}{\sum_{i=0}^{\lfloor \frac{d-1}{2} \rfloor} {n \choose i}(q-1)^i}.$$
A code that achieves the Hamming bound is called a perfect code. If $d=2t+1$ then it is called a perfect $t$-error correcting code. Since the covering radius of a perfect $t$-error correcting code is exactly $t$, from Lemma \ref{adG3}, we have the following result.

\begin{theorem} \label{MDG_perfect_a}
The minimum-distance graph of a perfect $t$-error correcting code is connected.
\end{theorem}

From Theorems \ref{aDG1} and \ref{MDG_perfect_a}, the following corollary holds.

\begin{corollary}\label{MDG_col1}
Let $f:\mathbb{F}_q^k \to \text{Im}(f)$ be a function. Then for an $(f\!:d_d,d_f)$-FCC with $d_f>d_d$, we have
$$
r_f(k:d_d,d_f) \ge n - k + 1,
$$
where $n$ is the integer satisfying
$
q^{n-k} = \sum_{i=0}^{\lfloor \frac{d_d-1}{2}\rfloor} {n \choose i} (q - 1)^i.
$
\end{corollary}

\subsubsection{MDS Codes}

The Singleton bound on the size of an error-correcting code $\mathcal{C}$ of length $n$ and minimum distance $d$ is given in \cite{H1986} as
\[
M \le q^{\,n-d+1},
\]
equivalently, for $k=\log_q M$,
$
d \le n-k+1.
$
A code that meets the Singleton bound with equality is called a \emph{maximum distance separable (MDS)} code. In particular, an $(n,M,d)_q$ code is MDS if and only if $M=q^{\,n-d+1}$.

To prove that the minimum-distance graph of an MDS code is connected we use the following lemma.

\begin{lemma}\label{lem:MDG_MDS}
Let $\mathcal{C}$ be an MDS code with parameters $(n,M,d)_q$, and let $u,v\in\mathcal{C}$. Then there exists $u'\in\mathcal{C}$ such that
\begin{enumerate}
    \item $d(u,u')=d$ (i.e., $u'$ is a neighbor of $u$ in $G(\mathcal{C})$);
    \item $d(u',v)\le d(u,v)-1$. 
\end{enumerate}
If $d(u,v)=d$, the codeword $v$ itself serves as $u'$.
\end{lemma}

\begin{proof}
Let $S=\{i\in[n]:u_i\ne v_i\}$ and $S^c=[n]\setminus S=\{i\in[n]:u_i=v_i\}$. Then $|S|=d(u,v)\ge d=n-k+1$, so $|S^c|=n-|S|\le k-1$. 
Choose a set $J\subseteq[n]$ with $|J|=k$ such that $S^c\subseteq J$ (i.e., $J$ contains all agreements of $u$ and $v$, and adds $k-|S^c|$ positions from $S$). Then
$
|J^c|=n-k=d-1,$ and 
$$|J\cap S|=k-|S^c|=k-(n-d(u,v))=d(u,v)-(d-1).
$$
Pick any $j\in J\cap S$ and define a vector $t\in\mathbb F_q^{J}$ by
$
t_j=v_j\ne u_j$ and $t_i=u_i\ \text{ for all } i\in J\setminus\{j\}.$
By the MDS projection property (for any $|J|=k$, the projection $\pi_J:\mathcal C\to \mathbb F_q^{J}$ is bijective), there exists a unique $u'\in\mathcal C$ with $\pi_J(u')=t$. Now, we have the following.
\begin{enumerate}
    \item \emph{$d(u,u')=d$.}  
The coordinates where $u$ and $u'$ may differ are: the single index $j\in J$ (forced to differ), plus indices in $J^c$ (unconstrained). Hence
$$
d(u,u')\le 1+|J^c| = 1+(d-1)=d.
$$
Since $u'\ne u$ and $d$ is the minimum distance of $\mathcal C$, we must have $d(u,u')=d$.

\item \emph{$d(u',v)\le d(u,v)-1$.}  
We count sure agreements between $u'$ and $v$. On $S^c\subseteq J$ we set $t_i=u_i=v_i$, so $u'_i=v_i$ for all $i\in S^c$; this gives $|S^c|=n-d(u,v)$ agreements. Additionally, at $j$ we enforced $u'_j=v_j$, yielding one more agreement. Therefore, the number of agreements is at least $n-d(u,v)+1$, and hence
\[
d(u',v)\le n-(n-d(u,v)+1) = d(u,v)-1.
\]
\end{enumerate}
This proves both claims.
\end{proof}

\begin{theorem}\label{MDG_MDS}
    The minimum-distance graph of any MDS code is connected.
\end{theorem}

\begin{proof}
Let $G(\mathcal C)$ be the minimum-distance graph of $\mathcal C$, and $u,v\in\mathcal C$. We will construct a path in $G(\mathcal{C})$ from vertex $u$ to $v$ by moving along edges while reducing the distance from $v$. Start with $x_0=u$, and apply Lemma~\ref{lem:MDG_MDS} to the pair $(x_0,v)$ to obtain $x_1\in\mathcal C$ with
$$
d(x_0,x_1)=d \quad\text{and}\quad d(x_1,v)\le d(x_0,v)-1.
$$
If $d(x_1,v)=0$, then we are already at $v$, otherwise $d(x_1,v)\geq d$. Again use Lemma \ref{lem:MDG_MDS} for the pair $(x_1,v)$, and obtain $x_s$. By repeating this process, we get a sequence $x_0,x_1,x_2,\dots$ in $\mathcal C$ such that
\[
d(x_i,x_{i+1})=d \qquad\text{and}\qquad d(x_{i+1},v)\le d(x_i,v)-1.
\]
Because $d(x_i,v)$ is a nonnegative integer that decreases at each step, after finitely many steps we reach some $x_j$ with $d(x_j,v)=d$. 
The next application of Lemma~\ref{lem:MDG_MDS} gives us $x_{j+1}=v$ with $d(x_j,v)=d$. 
Thus $u=x_0,x_1,\dots,x_j,v$ is a path in $G(\mathcal C)$ from $u$ to $v$. Therefore, $G(\mathcal C)$ is connected.
\end{proof}
From Theorems \ref{aDG1} and \ref{MDG_MDS}, the following corollary holds.

\begin{corollary}\label{col:MDG_MDS}
Let $f:\mathbb{F}_q^k \to \mathrm{Im}(f)$ be a function. 
Assume there exists an MDS $(n,q^k,d)_q$ code, i.e., $n=k+d-1$.
Then for an $(f\!:d,d_f)$-FCC with $d_f>d$, we have
\[
r_f(k:d,d_f) \ge n-k+1 = d.
\]
\end{corollary}

\subsection{Minimum-distance graph of linear codes}
In this subsection, we give a class of linear codes whose minimum-distance graph is connected.

\begin{lemma}\label{lem:LC}
Let $\mathcal{C}$ be a linear code over $\mathbb{F}_q$ with minimum distance $d$, and let
$W_d$ denote the set of all codewords of $\mathcal{C}$ having Hamming weight $d$.
Then $\mathcal{C}$ is spanned by its minimum-weight codewords, i.e.,
$\mathcal{C} = \langle W_d \rangle$, if and only if the minimum-distance graph
$G(\mathcal{C})$ is connected.
\end{lemma}

\begin{proof}
Since $\mathcal{C}$ is linear, we have
$d(u,v)=\mathrm{wt}(u-v),$ for all $u,v \in \mathcal{C}.$
Assume that $\mathcal{C} = \langle W_d \rangle$.
Let $u \in \mathcal{C}$, then there exist codewords
$w_1,w_2,\ldots,w_t \in W_d$ such that
\[
u = w_1 + w_2 + \cdots + w_t.
\]
Starting from the zero codeword and successively adding $w_1,w_2,\ldots,w_t$,
each step corresponds to adding a minimum-weight codeword and hence moves along an
edge of $G(\mathcal{C})$. After $t$ steps, we reach $u$.
Thus, every codeword is connected to the zero codeword by a path in $G(\mathcal{C})$.
Since $\mathcal{C}$ is linear, it follows that any two codewords are connected, and
therefore $G(\mathcal{C})$ is connected.

Conversely, suppose that $\mathcal{C} \neq \langle W_d \rangle$.
Let $\mathcal{C}' = \langle W_d \rangle$, which is a proper subcode of $\mathcal{C}$.
Now we show that $G(\mathcal{C})$ is disconnected.
If $x,y \in \mathcal{C}$ are adjacent in $G(\mathcal{C})$, then
\[
x-y \in W_d \subseteq \mathcal{C}',
\]
which implies that $x$ and $y$ belong to the same coset of $\mathcal{C}'$.
Consequently, there is no edge between vertices belonging to distinct cosets of
$\mathcal{C}'$, and $G(\mathcal{C})$ decomposes into components indexed by the cosets of $\mathcal{C}'$.
Hence, $G(\mathcal{C})$ is disconnected.
\end{proof}

From Theorem \ref{aDG1} and Lemma \ref{lem:LC}, the following corollary holds.

\begin{corollary}\label{col:LC}
Let $\mathcal{C}$ be an $[n,k,d]$ linear code over $\mathbb{F}_q$ that is spanned by its
minimum-weight codewords. Then $\mathcal{C}$ cannot be an $(f\!:d,d_f)$-FCC for any
$f:\mathbb{F}_q^k \to \mathrm{Im}(f)$ with $|\mathrm{Im}(f)| \ge 2$ and $d_f > d$.
\end{corollary}

\subsection{$\alpha$-distance graph of some other classes of codes}

There exist several classes of codes for which the covering radius is known.
Moreover, even an upper bound on the covering radius is sufficient for our
purposes. In particular, by Corollary~\ref{aDG4}, if the covering radius of a
code $\mathcal{C}$ satisfies $R(\mathcal{C}) \le x$, then $\mathcal{C}$ cannot be
an $(f\!:d,d_f)$-FCC for any function
$f:\mathbb{F}_q^k \to \mathrm{Im}(f)$ with $|\mathrm{Im}(f)| \ge 2$ and
$d_f > 2x+1$.

\begin{itemize}
\item \textbf{Quasi-perfect codes:}  
Let $\mathcal{C}_0 \subseteq \mathbb{F}_q^{\,n}$ be a perfect
$(n,q^k,2t+1)$ code. A code $\mathcal{C}\subseteq \mathbb{F}_q^{\,n+1}$
with $|\mathcal{C}|=|\mathcal{C}_0|=q^k$ is called a \emph{quasi-perfect code}
if puncturing $\mathcal{C}$ in one coordinate yields $\mathcal{C}_0$, i.e.,
there exists a coordinate $i$ such that
\[
\pi_i(\mathcal{C}) = \mathcal{C}_0,
\]
where $\pi_i$ denotes the puncturing map that deletes the $i$-th coordinate.
Classical examples include extended Hamming codes and extended Golay codes.
The covering radius of quasi-perfect codes is $t+1$, and well established
in the literature \cite{L1992}.

\item \textbf{First-order Reed-Muller codes:}  
For the binary first-order Reed-Muller code $\mathrm{RM}(1,m)$ of length $2^m$,
the covering radius is known exactly and is given by
\[
R(\mathrm{RM}(1,m)) = 2^{m-1} - 2^{\lceil m/2\rceil - 1}.
\]
This result was established in \cite{L1992}.
\end{itemize}

A general upper bound on the covering radius of Reed-Muller codes $\mathrm{RM}(r,m)$ is given in \cite{CL1994}.  In \cite{JM1999}, the authors showed that the covering radius of an $[n,k,d]$ linear code over $\mathbb{F}_q$ with dual distance
$d^{\perp}$ satisfies
\begin{equation*}
\label{eq:JanwaMattsonBound}
R(\mathcal{C}) \le 
n - \sum_{i=0}^{d^{\perp}-2}
\left\lceil \frac{n-i}{q} \right\rceil .
\end{equation*}

\section{Conclusion}
\label{conclusion}

This work established non-existence results for \emph{strict} $(f\!:\!d_d,d_f)$-function-correcting codes with data protection, where the function values are required to be protected more strongly than the data. By analyzing distance graphs associated with codes, we identified connectivity conditions that rule out the possibility of such strict FCCs.

The proposed graph-based approach allowed us to relate the feasibility of strict FCCs to classical parameters of error-correcting codes. In particular, we showed that optimal families such as perfect codes and MDS codes cannot admit strictly larger protection for function values than for data. These results highlight the inherent structural limitations of FCCs with data protection.

\section*{Acknowledgement}
This work was supported by a joint project grant to Aalto University and Chalmers University of Technology (PIs A. Graell i Amat and C. Hollanti) from the Wallenberg AI, Autonomous Systems and Software Program, and additionally by the Science and Engineering Research Board (SERB) of the Department of Science and Technology (DST), Government of India, through the J.C. Bose National Fellowship to Prof. B. Sundar Rajan.

\end{document}